\documentclass[runningheads,orivec]{llncs}

\usepackage{hyperref}
\usepackage[utf8]{inputenc}
\usepackage{amsmath,amssymb}
\usepackage{dsfont}
\usepackage{color}
\usepackage{graphicx}
\usepackage{url}
\usepackage{centernot}
\usepackage[english,ruled]{algorithm2e}
\usepackage{newtxtext,newtxmath}

\usepackage[mode=buildnew]{standalone}
\usepackage{tikz}
\usetikzlibrary{automata, positioning, arrows, arrows.meta}
\usetikzlibrary{decorations.markings, shapes.geometric, shapes.misc}
\usetikzlibrary{myautomata}
\usepackage{varwidth}

\usepackage{listings}
\lstdefinelanguage{its}
{morekeywords={if, =, +, >=, +, -, (, ), [, ], true, false, typedef, transition, int,GAL, abort, !, \{, \}, label, ", &&, ., composite, synchronization,for,gal, array},
morecomment=[l]{//},
sensitive=false,
}
\newcommand{\nat}%
{\ensuremath{\mathds{N}}}
\newcommand{\zrel}%
{\ensuremath{\mathds{Z}}}
\newcommand{\rea}%
{\ensuremath{\mathds{R}}}
\newcommand{\bool}%
{\ensuremath{\mathds{B}}}

\newcommand{\pre}%
{\ensuremath{\preccurlyeq}}

\usepackage{accents}

\newcounter{rrule}

\renewcommand{\implies}{\Rightarrow}

\newcommand{\lang}{\ensuremath{\mathscr{L}}}

\newcommand{\B}{\ensuremath{\mathds{B}}}
\newcommand{\AP}{\ensuremath{\mathit{AP}}}
\newcommand{\Acc}{\ensuremath{\mathit{Acc}}}
\newcommand{\Trim}{\ensuremath{\mathit{Trim}}}
\def\init{\iota}

\DeclareMathOperator{\SG}{\mathsf{SG}}
\DeclareMathOperator{\TG}{\mathsf{TG}}

\newcommand{\F}{\mathsf{F}} 
\newcommand{\G}{\mathsf{G}} 
\newcommand{\X}{\mathsf{X}} 

\newcommand{\givenmin}[2]{\ensuremath{{\min}_{|#2}(#1)}}
\newcommand{\givenmax}[2]{\ensuremath{{\max}_{|#2}(#1)}}
\newcommand{\givenminqe}[2]{\ensuremath{{\min}^{\exists}_{|#2}(#1)}}
\newcommand{\givenmaxqe}[2]{\ensuremath{{\max}^{\exists}_{|#2}(#1)}}
\newcommand{\givenminato}[2]{\ensuremath{\mathit{BM}_{|#2}(#1)}}
\newcommand{\givensirelax}[2]{\ensuremath{\mathit{sirelax}_{|#2}(#1)}}
\newcommand{\givensirestrict}[2]{\ensuremath{\mathit{sirestrict}_{|#2}(#1)}}
\newcommand{\si}{\mathit{si}}
\renewcommand{\ss}{\mathit{ss}}

\definecolor{lime}{HTML}{A6CE39}
\DeclareRobustCommand{\orcidicon}{
	\hspace{-2.5mm}
	\begin{tikzpicture}[baseline={(0,-0.12)}]
	\draw[lime, fill=lime] (0,0)
	circle [radius=0.16]
	node[white] (ID) {{\fontfamily{qag}\selectfont \tiny ID}};
	\draw[white, fill=white] (-0.0625,0.095)
	circle [radius=0.007];
	\end{tikzpicture}
	\hspace{-2.5mm}
      }
\def\orcidID#1{\href{https://orcid.org/#1}{\smash{\orcidicon}}}

\usepackage{booktabs}
\usepackage{subcaption}
\usepackage{graphicx}

\usepackage{csvsimple}
\usepackage{longtable}

\usepackage{varwidth}
\usepackage{amssymb, amsmath, thmtools, thm-restate}
\usepackage{ifxetex,ifluatex}
\usepackage{longtable}
\usepackage{array}
\usepackage{multirow}
\usepackage{colortbl}
\usepackage{makecell}
\usepackage{xcolor}
\usepackage{graphicx}
\usepackage{standalone}
\usepackage{textcomp}
\usepackage[misc,geometry]{ifsym}

\begin{document}
\title{Simplifying LTL Model Checking\\ Given Prior Knowledge}



%
%

\author{Alexandre Duret-Lutz\inst{3}\orcidID{0000-0002-6623-2512}, Denis Poitrenaud\inst{1,2}\orcidID{0009-0007-5038-7804}, Yann Thierry-Mieg\Letter\inst{1}\orcidID{0000-0001-7775-1978}}

%

\authorrunning{A.~Duret-Lutz, D.~Poitrenaud, Y.~Thierry-Mieg}

%
\institute{Sorbonne Université, CNRS, LIP6, F-75005 Paris, France \\
\email{denis.poitrenaud@lip6.fr, yann.thierry-mieg@lip6.fr}
\and Université Paris Cité, F-75006 Paris, France \\
 \and EPITA, LRE, Le Kremlin-Bicêtre, France\\
\email{adl@lrde.epita.fr}
}

\maketitle              
\begin{abstract}
  We consider the problem of the verification of an LTL specification
  $\varphi$ on a system $S$ given some prior knowledge $K$, an LTL
  formula that $S$ is known to satisfy.  The automata-theoretic
  approach to LTL model checking is implemented as an emptiness check
  of the product $S\otimes A_{\lnot\varphi}$ where $A_{\lnot\varphi}$
  is an automaton for the negation of the property.  We propose new
  operations that simplify an automaton $A_{\lnot\varphi}$
  \emph{given} some knowledge automaton $A_K$, to produce an automaton
  $B$ that can be used instead of $A_{\lnot\varphi}$ for more
  efficient model checking.

  Our evaluation of these operations on a large benchmark derived from
  the MCC'22 competition shows that even with simple knowledge, half
  of the problems can be definitely answered without running an LTL
  model checker, and the remaining problems can be simplified
  significantly.
\end{abstract}

\section{Introduction --- Knowledge is Power}

LTL model checking consists in verifying whether all infinite
executions of a system $S$ satisfy an LTL formula $\varphi$, \emph{i.e.},
$\lang(S)\subseteq\lang(\varphi)$.  In this case we write
$S\models \varphi$.  In the automata-theoretic approach to model
checking~\cite{Vardi07}, this inclusion test is usually implemented as
an emptiness check of the product of two automata:
$\lang(S\otimes A_{\lnot\varphi})=\emptyset$, where $A_{\lnot\varphi}$
represents the negation of $\varphi$.

The premise of this paper is that we assume to have some additional
knowledge $K$ about $S$.  In particular, the knowledge we consider are
over-approximations of the system: $\lang(S)\subseteq\lang(K)$.  For
instance $K$ might be an LTL formula that has already been proven on
$S$.
Of course if $K$ implies $\varphi$, \emph{i.e.} $\lang(K) \subseteq \lang(\varphi)$, then $\varphi$
 holds as well since $\lang(S) \subseteq \lang(K)$.
 And if $\lang(K) \subseteq \lang(\lnot\varphi)$,
 any run of the system is a counter-example.

 But if none of these basic implications hold, we can still benefit
 from prior knowledge.  We show that verifying $S\models \varphi$
 \emph{given $K$} is equivalent to checking
 $\lang(S\otimes B)=\emptyset$ for an automaton $B$ that is
 \emph{simpler} than $A_{\lnot\varphi}$, hopefully allowing a faster
 exploration of $S\otimes B$.

 As an example the automaton $A_{\lnot\varphi}$ that is on the left of
 Figure~\ref{fig:minato} (page~\pageref{fig:minato}) can be replaced
 by the automaton $B$ that is on the right of the same figure.  This
 new automaton is smaller, uses fewer atomic propositions, is now
 deterministic, and needs fewer acceptance sets because it is now a
 terminal automaton~\cite{cerna.03.mfcs,manna.90.podc}.  Using this
 automaton $B$ should therefore simplify the job of a model checker.

This paper is organized as follows.  In Section~\ref{sec:lang} we
formalize notion of $S\models \varphi$ \emph{given $K$} from the point
of view of languages, and discuss possible goals when transposing this
on automata.  In Section~\ref{sec:aut} we pose useful definitions, then
Section~\ref{sec:basic} proposes basic and Section~\ref{sec:bounds} advanced automata operations
that aim to simplify $A_{\lnot\varphi}$ based on some given knowledge
$K$.  In Section~\ref{sec:stuttu} we propose costlier automata operations
that aim to modify $A_{\lnot\varphi}$ to make it
stutter-insensitive, within the bounds allowed by some knowledge $K$.
Finally, in Section~\ref{sec:bench} we evaluate the
above techniques on a large third-party benchmark provided by the model checking contest~\cite{kordon.21.sttt}.



\section{Bounding Languages ``Given That...''}\label{sec:lang}

In this section, we focus on providing justification for our approach
at the \emph{language} level.  The language {\lang(X)} of a system or
property $X$ is a set of infinite words over an alphabet $\Sigma$,
$\lang(X) \subseteq \Sigma^\omega$.  We denote
$\overline{\lang(X)}=\Sigma^\omega\setminus \lang(X)$ the complement
of the language of $X$.

A system $S$ satisfies property $\varphi$, denoted $S\models \varphi$
if and only if the language $\lang(S)$ of the system is a subset of
the property language $\lang(\varphi)$, \emph{i.e.},
$\lang(S) \subseteq \lang(\varphi)$.  When $\varphi$ is an LTL
formula, the classical automaton-based approach~\cite{Vardi07} is to
test $\lang(S) \cap \lang{(\lnot\varphi)} = \emptyset$, \emph{i.e.},
perform an emptiness check with the language of the negated property.

In the following, we assume that $\lang(S)\ne\emptyset$ since the
empty system would satisfy any property and its negation.

Now, consider a property $K$ (a \emph{knowledge}) such that it has
already been established that $S\models K$, \emph{i.e.}, we know that
$\lang(S) \subseteq \lang(K)$.  This \emph{a priori} knowledge gives
us some degrees of freedom when testing whether $S$ satisfies a new
property $\varphi$. Indeed, we already know that words outside
$\lang(K)$ are definitely \emph{not} part of $\lang(S)$.

The main intuition is given by Fig.~\ref{fig:patate}.  Since
$\lang(S)\subseteq\lang(K)$, it is safe to replace the test
$\lang(S) \cap \lang{(\lnot\varphi)} = \emptyset$ by a test
$\lang(S) \cap \lang(B) = \emptyset$ where $\lang(B)$ is built from
$\lang{(\lnot\varphi)}$ by either removing or including words of
$\overline{\lang(K)}$.  Indeed, words in $\overline{\lang(K)}$ are not
part of the system, so they cannot belong to
$\lang(S) \cap \lang(\lnot \varphi)$.

\begin{figure}[tb]
    \centering
    \begin{subfigure}{0.32\textwidth}
        \centering
  \begin{tikzpicture}
    \draw[draw=magenta] (0,1) ellipse[x radius=1cm,y radius=2cm];
    \node at (0,1.9) {$\lang(\lnot\varphi)$};
    \begin{scope}
      \clip (1,0.2) ellipse[x radius=1.8cm,y radius=1.1cm];
      \fill[fill=magenta!30,draw=magenta] (0,1) ellipse[x radius=1cm,y radius=2cm];
    \end{scope}
    \draw[draw=cyan] (1,0.2) ellipse[x radius=1.8cm,y radius=1.1cm];
    \draw[draw=gray,dashed] (1,-0.3) ellipse[x radius=.8cm,y radius=.5cm] node[xshift=2mm]{$\lang(S)$};
    \node at (2,0.2) {$\lang(K)$};
    \node at (0,0.2) {$\lang(B)$};
    \draw (current bounding box.north west) -- (current bounding box.north east)
    -- (current bounding box.south east) -- (current bounding box.south west)
    -- cycle;
  \end{tikzpicture}
          \caption{The most restricted $\lang(B)$ that can be constructed from $K$, $\lang(B)=\lang(\lnot\varphi) \cap \lang(K)$}
        \label{fig:restrict}
    \end{subfigure}
    \hfill
    \begin{subfigure}{0.32\textwidth}
        \centering
   \begin{tikzpicture}
    \fill[fill=magenta!30,draw=magenta] (0,1) ellipse[x radius=1cm,y radius=2cm];
    \node[align=center] at (0,1.9) {$\lang(B)$\\=\\$\lang(\lnot\varphi)$};
    \draw[draw=cyan] (1,0.2) ellipse[x radius=1.8cm,y radius=1.1cm];
    \node at (2,0.2) {$\lang(K)$};
    \draw[draw=gray,dashed] (1,-0.3) ellipse[x radius=.8cm,y radius=.5cm] node[xshift=2mm]{$\lang(S)$};
    \draw (current bounding box.north west) -- (current bounding box.north east)
    -- (current bounding box.south east) -- (current bounding box.south west)
    -- cycle;
  \end{tikzpicture}
         \caption{Classic approach simply using the language of $\lnot\varphi$, $\lang(B)=\lang(\lnot\varphi)$.}
        \label{fig:product}
    \end{subfigure}
    \hfill
    \begin{subfigure}{0.32\textwidth}
        \centering
   \begin{tikzpicture}
    \fill[fill=magenta!30,draw=magenta] (0,1) ellipse[x radius=1cm,y radius=2cm];
    \draw[draw=cyan] (1,0.2) ellipse[x radius=1.8cm,y radius=1.1cm];
    \fill[fill=magenta!30] (current bounding box.north west) -- (current bounding box.north east)
    -- (current bounding box.south east) -- (current bounding box.south west)
    -- cycle;
    \draw[draw=cyan,fill=white] (1,0.2) ellipse[x radius=1.8cm,y radius=1.1cm];
    \fill[fill=magenta!30,draw=magenta] (0,1) ellipse[x radius=1cm,y radius=2cm];
    \draw[draw=cyan] (1,0.2) ellipse[x radius=1.8cm,y radius=1.1cm];
    \node at (0,1.9) {$\lang(\lnot\varphi)$};
    \node at (2,0.2) {$\lang(K)$};
    \node at (2,1.9) {$\lang(B)$};
    \draw[draw=gray,dashed] (1,-0.3) ellipse[x radius=.8cm,y radius=.5cm] node[xshift=2mm]{$\lang(S)$};
    \draw (current bounding box.north west) -- (current bounding box.north east)
    -- (current bounding box.south east) -- (current bounding box.south west)
    -- cycle;
  \end{tikzpicture}
        \caption{The most relaxed $\lang(B)$ that can be constructed given $K$, $\lang(B)= \lang(\lnot\varphi) \cup \overline{\lang(K)}$.}
        \label{fig:relax}
    \end{subfigure}

    \caption{The outside box represents all words in $\Sigma^\omega$. Each language is depicted as an ellipse, with the language of the system $\lang(S)$ inside the knowledge $\lang(K)$ but we do not know
    whether the system language overlaps the negated property language $\lang(\lnot\varphi)$. The language $\lang(B)$ represented in magenta can be chosen anywhere between these extremes to replace $\lang(\lnot\varphi)$ in the model-checking procedure.}
    \label{fig:patate}
\end{figure}
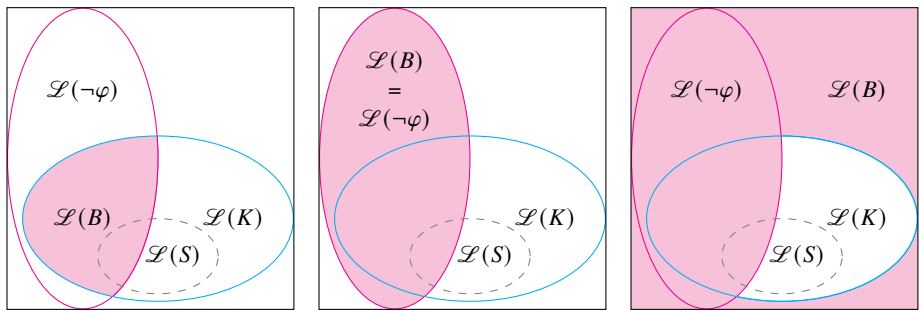

This leads to the following theorem whose proof follows immediately
from Figure~\ref{fig:patate}.

\begin{theorem}\label{th:rel}
  \label{knowledge-theorem}
  Let $S$ be a system, and $K$ a property such that
  $\lang(S) \subseteq \lang(K)$.  For any property $\varphi$, we can
  define a $\lang(B)$ such that
  $\lang(S) \cap \lang(\lnot\varphi) = \emptyset$ if and only if
  $\lang(S) \cap \lang(B) = \emptyset$, by choosing $\lang(B)$
  between the following bounds:

  \[
    \lang(\lnot\varphi) \cap \lang(K) \quad\subseteq\quad \lang(B) \quad\subseteq\quad \lang(\lnot \varphi)\cup \overline{\lang(K)}
    \]
\end{theorem}

The lower bound $\lang(\lnot\varphi) \cap \lang(K)$ is called the
\emph{restriction} of $\lnot\varphi$: it is constructed from
$\lang(\lnot\varphi)$ by removing words from $\overline{\lang(K)}$
(Fig.~\ref{fig:restrict}).  The upper bound
$\lang(\lnot \varphi) \cup \overline{\lang(K)}$ is the
\emph{relaxation} of $\lnot\varphi$, constructed from
$\lang(\lnot\varphi)$ by adding words from $\overline{\lang(K)}$
(Fig.~\ref{fig:relax}).

%

The above theorem gives us more freedom in the automata-theoretic
approach to LTL model checking.  In this context, both the property
$\varphi$ and the knowledge $K$ are expressed as linear-time temporal
logic (LTL) formulas which can be converted into automata over
infinite words.

Thus, model checking $S\models\varphi$ is implemented as
$\lang(S\otimes A_{\lnot\varphi}) = \emptyset$ where
$A_{\lnot\varphi}$ is an automaton for $\lnot \varphi$, and $\otimes$
is the product of automata~\cite{Vardi07}.  Here, we want to use
Theorem~\ref{th:rel} to find a \emph{simpler} automaton $B$ such that
model checking with $\lang(S\otimes B) = \emptyset$ is more efficient.
Contrary to intuition, choosing the automaton $B$ with the smallest
\emph{language} might be counter-productive, because a small language
does not necessarily equate to a small automaton.

\label{sec:goals}

To make model checking more efficient, we target the
following goals:
\begin{description}
\item[smaller or more deterministic] Reducing the size of $B$, or
  making it more deterministic can often reduce the size of the product
  $S\otimes B$. (Blahoudek et al.~\cite{blahoudek.14.spin} suggest that
  contrary to previous measurements~\cite{sebastiani.03.charme}, ``smaller'' is
  more important than ``more deterministic'' for model checking.)
\item[simpler strength class] The emptiness check algorithms can be
  simplified if $B$ belongs to simpler classes of automata, such as
  \emph{weak}, or \emph{terminal}
  automata~\cite{bloem.99.cav,cerna.03.mfcs,renault.13.tacas}.
\item[stutter-insensitive] For concurrent systems, many partial-order
  reductions (POR)
  techniques~\cite{peled.94.cav,valmari.93.cav,godefroid.96.phdlncs}
  and structural reductions~\cite{haddad.06.ppl,YTM20} can be used
  when it is known that $B$ is stutter-insensitive.
\item[fewer atomic proposition checks] Reducing the number of atomic
  propositions and the syntactic complexity of the
  formulas labeling the edges of $B$ can reduce the time required to
  build $S\otimes B$ in explicit model
  checking~\cite{blahoudek.15.spin}, and reducing the set of observed
  propositions also helps the aforementioned POR based techniques.
\end{description}

The techniques we will propose mainly attempt to reduce the size of
the automata, their number of atomic propositions, and attempt to make
them stutter-insensitive.  Any determinism improvement or strength
reduction is a welcome side effect.

\section{Simplifying Automata ``Given That...''}\label{sec:aut}

We now turn the language bounds of Section~\ref{sec:lang} into
automata constructions.  We use a variant of Büchi automata called
\emph{transition-based generalized Büchi automata} (TGBA).  This
variant uses accepting transitions instead of accepting states.
Additionally, the acceptance condition is generalized: a run has to
visit \emph{multiple} accepting sets of transitions infinitely often.
This variant is particularly compact to express weak fairness
conditions~\cite{couvreur.99.fm}, and it also makes our subsequent
definitions easier without loss of generality.

\subsection{Definitions}\label{sec:def}

The following definitions are freely adapted from the literature.

Let $\B=\{\bot,\top\}$ represent the Boolean set, and let $\AP$
represent a set of Boolean atomic propositions.  A valuation $\ell$ is
a function from $\AP$ to $\B$.  The set of valuations is denoted
$\B^\AP$.  The set of Boolean formulas over $\AP$ is denoted
$\B(\AP)$.  In the following, we consider words that are infinite
sequences of valuations, therefore the alphabet $\Sigma$ of
Section~\ref{sec:lang} is $\Sigma=\B^\AP$.  For an atomic proposition
$a\in\AP$ we use $\bar a$ or $\lnot a$ interchangeably to represent
its negation.

\begin{definition}[TGBA]
  A Transition-based Generalized Büchi Automaton (TGBA), is a structure
  $A=\langle \AP, Q, \init, \Acc, \delta\rangle$ where
  \begin{itemize}
  \item $\AP$ is a finite set of Boolean atomic propositions,
  \item $Q$ is a finite set of states,
    \item $\init\in Q$ is the initial state,
    \item $\Acc$ is a finite set of acceptance marks (denoted \emph{\tacc0}, \emph{\tacc1}, \emph{\tacc2}, etc.)
    \item $\delta\subseteq Q\times \B(\AP)\times 2^\Acc \times Q$ is the transition relation
      where we use $t=q\xrightarrow{f,a}q'$ to denote an element $t \in \delta$, $f$ is a Boolean formula over $\AP$ that we call the \emph{label} of the transition and $a$ is a set of acceptance marks.
  \end{itemize}

  A \emph{run} of $A$ on an infinite word
  $w=\ell_1\ell_2\ell_3\ldots\in(\B^\AP)^\omega$ is an infinite
  sequence of connected transitions
  $\rho=q_1\xrightarrow{f_1,a_1}q_2\xrightarrow{f_2,a_2}q_3\xrightarrow{f_3,a_3}q_4\ldots\in
  \delta^\omega$ such that $q_1=\init$ and for all $i$,
  $\ell_i\Rightarrow f_i$.  (Recall that $\ell_i$ is a valuation of
  all atomic propositions, therefore a conjunction of atomic
  propositions, in negative or positive form, but $f_i$ is a Boolean
  formula.)
  A run is \emph{accepting} iff for each mark $m\in \Acc$ there are
  infinitely many $i$ such that $m\in a_i$.

  The \emph{language} of $A$, denoted $\lang(A)$, is the set of
  all infinite words $w$ such that there exists an accepting run of
  $A$ on $w$.

\end{definition}

\begin{theorem}[TGBA for a formula~\cite{couvreur.99.fm,gastin.01.cav,giannakopoulou.02.forte}]
  Given an LTL formula $\varphi$ over $\AP$, one can build a TGBA $A_\varphi$
  with $O(2^{|\varphi|})$ states such that $\lang(A_\varphi)=\lang(\varphi)$.
\end{theorem}

For instance the leftmost automaton of
Figure~\ref{fig:minato}~(page~\pageref{fig:minato}) is a TGBA for
$\F(p\land r)\lor \G((\F q)\lor (\F\bar q))$.  An accepting run has to
encounter marks \tacc{0} and \tacc{1} infinitely often.  Therefore, any
run reaching state 1 is accepting, and any run
reaching state 2 is accepting if both $q$ and $\bar q$ hold
infinitely often.  Note that $A_\varphi$ is not unique.  There is a
vast literature on techniques for building and simplifying
automata~\cite[...]{etessami.00.concur,somenzi.00.cav,etessami.01.alp,fritz.03.ciaa,babiak.13.spin,duret.14.ijccbs}.

An obvious optimization is to discard the useless parts of the automaton by trimming it.
%
  The \emph{trim} of an automaton $A$, denoted $\Trim(A)$, is the
  restriction of $A$ to the transitions and states that appear in at
  least one accepting run of $A$.  Doing so preserves the language of $A$.
  This operation can be done in linear time by studying the strongly
  connected components of the automaton~\cite{etessami.00.concur}.
%

The intersection of the languages of two automata $A_1$ and $A_2$ is
represented by a product $A_1\otimes A_2$ such that
$\lang(A_1\otimes A_2)=\lang(A_1)\cap\lang(A_2)$.

\begin{definition}[Product of TGBA]
  Given two automata $A_1=\langle \AP_1, Q_1, \init_1, \Acc_1, \delta_1\rangle$
  and $A_2=\langle \AP_2, Q_2, \init_2, \Acc_2, \delta_2\rangle$,
  where $\Acc_1\cap\Acc_2 = \emptyset$,
  the product $A_1\otimes A_2$ is the automaton
  $\langle \AP, Q, \init, \Acc, \delta \rangle$ where:
  \vspace*{-1ex}
  \begin{itemize}
  \item $\AP = \AP_1 \cup \AP_2$
  \item $Q = Q_1\times Q_2$
  \item $\init = (\init_1,\init_2)$
  \item $\Acc = \Acc_1 \cup \Acc_2$.
  \item $\delta = \left\{ (q_1,q_2)\xrightarrow{f_1\land f_2,a_1\cup a_2} (q'_1,q'_2) \;\middle|\; q_1\xrightarrow{f_1,a_1} q'_1 \in \delta_1, q_2\xrightarrow{f_2,a_2} q'_2\in \delta_2\right\}$
  \end{itemize}
\end{definition}

For instance Figure~\ref{fig:construct} shows in the bottom right the
product $A_{\lnot\varphi}\otimes A_K$ of the two surrounding
automata.  The transitions that would be removed by $\Trim$ are dashed.

One can also define the sum of two TGBA $A_1\oplus A_2$ such that
$\lang(A_1\oplus A_2)=\lang(A_1)\cup\lang(A_2)$, and the complement
$\overline{A}$ such that
$\lang(\overline{A})=(\B^\AP)^\omega\setminus \lang(A)$.  We omit the
precise definition of these operations.  While sum and product are
cheap operations (at most quadratic in the size of the automata), the
complement is worse than exponential~\cite{yan.08.lmcs,schewe.12.atva}
so it is often desirable to avoid it (e.g., we prefer to compute
$A_{\lnot\varphi}$ instead of $\overline{A_{\varphi}}$).

\section{Basic Strategies}
\label{sec:basic}

The simplest way to apply Theorem~\ref{th:rel} is to build automata
for the most restricted and the most relaxed languages pictured in Figure~\ref{fig:patate}.  Consider
the following two definitions:
\begin{align}
  \givenmin{A_{\lnot\varphi}}{K}&=A_{\lnot\varphi}\otimes A_K \\
  \givenmax{A_{\lnot\varphi}}{K}&=A_{\lnot\varphi}\oplus A_{\lnot K}
\end{align}

When $B$ is chosen as $\givenmin{A_{\lnot\varphi}}{K}$, we are using
the most restricted language of Figure~\ref{fig:restrict}.  If $B$ is
$\givenmax{A_{\lnot\varphi}}{K}$ we are using the most
relaxed language of Figure~\ref{fig:relax}.

Note that if $\lang(\givenmin{A_{\lnot\varphi}}{K})=\emptyset$, then
it follows from the definition that $\lang(K)\subseteq\lang(\varphi)$,
and since $\lang(S)\subseteq\lang(K)$ we have $S\models \varphi$.
Dually, if $\lang(\givenmax{A_{\lnot\varphi}}{K})=(\B^\AP)^\omega$,
then $\lang(K)\subseteq(\lnot\varphi)$, which means that
$S\models \lnot\varphi$ (i.e., every run of $S$ is a counterexample of
$\varphi$) and therefore $S\not\models \varphi$ (because $S$ is
nonempty).

While the emptiness check of a TGBA $\lang(A)=\emptyset$ can be
performed in linear time~\cite{couvreur.99.fm}, the universality test
$\lang(A)=(\B^\AP)^\omega$ requires exponential
time~\cite{fogarty.10.tacas}.  Fortunately the universality test
$\lang(\givenmax{A_{\lnot\varphi}}{K})=(\B^\AP)^\omega$ can be
avoided by replacing it with
$\lang(\givenmin{A_{\varphi}}{K})=\emptyset$ provided a formula for $\varphi$ is known.

Moreover, the automata products and sums in the above $\min_{|K}$ and $\max_{|K}$ constructions can also be replaced by logical
 operations on formulas before translating them to TGBA, as in $A_{\lnot\varphi\land K}$ and $A_{\lnot\varphi\lor \lnot K}$ respectively.

In the case where the $\min$ and $\max$ automata are neither empty nor
universal, their sizes are unlikely to be smaller than the original
$A_{\lnot\varphi}$.  In a way, using these automata for model checking
is similar to asking the model checker to prove $K$ in addition to
$\lnot\varphi$.  As stated in Section~\ref{sec:goals}, we would prefer
to select a $B$ that is ``simpler'' than $A_{\lnot\varphi}$.

The knowledge $K$ could contain atomic propositions that do not
appear in $\lnot\varphi$.  Let $P$ be the set of atomic propositions
that appear in $K$ but not in $\varphi$.
To avoid introducing needless atomic propositions in $P$, we suggest
to existentially quantify them.  This quantification can be done
precisely on the automaton $A_K$ by existentially quantifying $P$ from
all labels, or it can be \emph{over approximated} on the LTL formula
$K$ by quantifying $P$ from all its Boolean subformulas (considered
individually).  We note $QE(P,K)$ the latter operation.  To show that
this is an over-approximation, consider the unsatisfiable formula
$K=\X(a\land b)\land \X(\bar a\land b)$ and $P=\{a\}$.  We have
$(\exists a,\,a\land b) = b$ and $(\exists a,\,\bar a\land b) = b$,
therefore, $QE(P,K)=\X(b)\land \X(b)=\X(b)$ which is satisfiable.


Assuming $P$ contains the atomic propositions of $K$ that are not in
$\varphi$, let us introduce the following notations:\\[-2em]
\begin{align}
  \givenminqe{\lnot\varphi}{K}&=A_{(\lnot\varphi)\land QE(P,K)} \label{eq:minqe}\\
  \givenmaxqe{\lnot\varphi}{K}&=A_{(\lnot\varphi)\lor\lnot QE(P,K)}\label{eq:maxqe}
\end{align}


\section{Using Transition-Based Boolean Bounds on Labels}
\label{sec:bounds}

In this section, we investigate how to leverage
theorem~\ref{knowledge-theorem} so that given an automaton for a
knowledge $K$, we rewrite the automaton $A_{\lnot\varphi}$ into a
 simpler automaton $B$.

Simplicity here is measured syntactically on the automaton; we want an
automaton that has fewer states, fewer transitions, fewer atomic
propositions, fewer acceptance marks, and simpler (smaller) Boolean
formulas labeling the transitions of the automaton.

To achieve this, we propose to compute a set of Boolean bounds for
each transition of the automaton $A_{\lnot\varphi}$.  These bounds
enable more flexibility in the selection of transition labels by
providing the most restrictive and the most relaxed Boolean formulas
that can label each transition.

Minato's algorithm~\cite{minato.92.sasimi} is a recursive way to
rewrite a Boolean formula as a prime-irredundant cover, which is very
compact in general.  The algorithm works recursively using formulas in
three-valued logic, and Minato~\cite[Section 4.4]{minato.92.sasimi}
suggests an implementation of this algorithm using Binary Decision
Diagrams~\cite{bryant.86.tc} where a three-valued formula is simply
bounded using two Boolean functions:
$(f_{\mathit{low}},f_{\mathit{high}})$ and the algorithm generates an
irredundant sum-of-product $f'$ such that
$f_{\mathit{low}}\implies f' \implies f_{\mathit{high}}$.  In other
words, $f'$ is generated as a disjunction of conjunctions of literals,
such that no conjunct is uncessary, and no literal can be removed from
any conjunct.  We use this algorithm to simplify transition labels, as
it removes literals that are unnecessary to stay within those bounds.

In Sections~\ref{sec:bub} and \ref{sec:blb}, we introduce strategies
to compute Boolean lower and upper bounds for each label of the
automaton.  Then, in Section~\ref{sec:minato}, we show how to simplify
transition labels by using Minato's algorithm on the computed bounds.  This
approach preserves the transition structure of the automaton.  It can
sometimes remove transitions (if its label becomes $\bot$), it can
remove states (when they become unreachable), it can reduce the number
of atomic propositions used, and it generally simplifies the
expression of the labels.  So, contrary to the
$\givenminqe{\lnot\varphi}{K}$ and $\givenmaxqe{\lnot\varphi}{K}$
approaches presented in Section~\ref{sec:basic}, this approach always
produces a simpler automaton.

\subsection{Boolean upper bounds}\label{sec:bub}

The first step consists in realizing that since $S \models K$, in
\emph{every state} of $A_K$ we are over-approximating the state
the system $S$ might be in.  Some paths in $A_K$ might not be
realizable by $S$, but the system definitely cannot do anything
that $K$ does not allow.

We start by building the synchronized product
$A_{\lnot\varphi} \otimes A_K$ in which every state is a pair $(q,k)$.
We can then apply the $\Trim$ operation to discard any transition that
does not belong to an accepting Strongly Connected Component (SCC) or
to the prefix of one, and then discard any state of the product
unreachable from the initial state.

 Now consider for a given state $q$ of $A_{\lnot \varphi}$ the set of
 states $Q_q$ of the knowledge automaton $A_K$ in correspondence with $q$.
 The state of the system $S$ in this set of states can be
 over-approximated as the logical disjunction of the formulas labeling
 any transition that is outgoing from any state in $Q_q$.

\begin{definition}[Knowledge-based state guarantee]\label{def:sg}
  Given two automata
  $A_{\lnot\varphi}=\langle \AP, Q,\linebreak[2] \init, \Acc , \delta\rangle$, and
  $A_K=\langle \AP, Q_K, \init_K, \delta_K, \Acc_K\rangle$, let
  $\Trim(A_{\lnot \varphi}\otimes A_K)=\langle \AP, Q_P, \init_P,
  \linebreak[2] \Acc_P, \delta_P\rangle$ be the trim product of
  $A_{\lnot \varphi}$ and $A_K$.

  For any state $q\in Q$, let $Q_q\subseteq Q_K$ denotes the subset of
  states of $A_K$ that can synchronize with $q$ in the product:
  \[
    Q_q = \{k \in Q_K \mid (q,k) \in  Q_P \}
  \]

  From this set of states, we define the \emph{state guarantee} in state $q$ :
  \[
    \SG(q) = \bigvee_{k \in Q_q} \bigvee_{k\xrightarrow{f,a}k' \in \delta_K} f
  \]
  that represents the disjunction of all transition labels of $A_K$ that
  leave a state of $Q_q$.
\end{definition}

\tikzset{tg/.style={},
         sg/.style={}}

\begin{figure}[ptb]
  \begin{tikzpicture}[automaton,semithick,initial overlay]
    \matrix (m) [matrix of math nodes, nodes={state}, column sep=1.33cm, row sep=1.33cm]
      {
          &$q_1$&|[initial below]|$q_0$&$q_2$\\
        |[initial]|$k_0$&$q_1,k_0$&|[initial above]|$q_0,k_0$&|[dashed]|$q_2,k_0$\\
        $k_1$&$q_1,k_1$&$q_0,k_1$&|[dashed]|$k_2,q_1$\\
      };
      \draw[->] (m-1-3) edge node[above]{$a\land c$} (m-1-2)
                (m-1-3) edge node[above]{$\bar a\lor \bar c$} (m-1-4)
                (m-1-2) edge[loop above,looseness=11] node{$\top$} pic[pos=.2]{acc=0} pic[pos=.8]{acc=1} (m-1-2)
                (m-1-3) edge[loop above] node{$\bar a\lor \bar c$} (m-1-3)
                (m-1-4) edge[loop above] node{$b$} pic[pos=.3]{acc=0} (m-1-4)
                (m-1-4) edge[loop right] node{$\bar b$} pic[pos=.3]{acc=1} (m-1-4);
      \draw[->] (m-2-1) edge node[left]{$b\land c$} (m-3-1)
                (m-2-1) edge[loop above] node[above]{$c$} (m-2-1)
                (m-3-1) edge[loop below] node[below]{$b\land c$} pic[pos=.3]{acc=2} (m-3-1);
      \draw[->] (m-2-3) edge[sg] node[above]{$a\land c$} (m-2-2)
                (m-2-3) edge[dashed] node[above]{$\bar a\land c$} (m-2-4)
                (m-2-3) edge[sg,out=230,in=200,loop] node[pos=.65,left]{$\bar a\land c$} (m-2-3)
                (m-2-2) edge[loop above,looseness=11,tg] node[above=-2pt]{$c$} pic[pos=.2]{acc=0} pic[pos=.8]{acc=1} (m-2-2)
                (m-2-4) edge[loop above,dashed] node[above]{$b\land c$}  pic[pos=.3]{acc=0} (m-2-4)
                (m-2-4) edge[loop right,dashed] node[right]{$\bar b\land c$}  pic[pos=.3]{acc=1} (m-2-4)
                (m-3-2) edge[out=300,in=240,loop,looseness=6,tg] node[below=2pt]{$b\land c$} pic[pos=.1]{acc=0} pic[pos=.38]{acc=1} pic[pos=.7]{acc=2} (m-3-2)
                (m-3-3) edge[sg,loop below] node{$\bar a\land b\land c$} pic[pos=.3]{acc=2} (m-3-3)
                (m-3-4) edge[loop below,looseness=11,dashed] node{$b\land c$}  pic[pos=.2]{acc=0} pic[pos=.8]{acc=2} (m-3-4)
                (m-2-2) edge[tg] node[left]{$b\land c$} (m-3-2)
                (m-2-3) edge[sg] node[right]{$\bar a\land b\land c$} (m-3-3)
                (m-2-4) edge[dashed] node[right]{$b\land c$} (m-3-4)
                (m-3-3) edge[sg] node[pos=.45,above]{$a\land b\land c$} (m-3-2)
                (m-3-3) edge[dashed] node[above]{$\bar a\land b\land c$} (m-3-4)
                (m-2-3) edge[sg,out=250,in=45] node[above=-2pt,pos=.6,sloped]{$a\land b\land c$} (m-3-2)
                (m-2-3) edge[out=340,in=115,dashed] node[above=-2pt,pos=.55,sloped]{$\bar a\land b\land c$}(m-3-4)
                ;


   \node [left=1mm of m-1-2] {$(A_{\lnot\varphi})$};
   \node [left=5mm of m-2-1] {$(A_K)$};
   \node [right=1.5cm of m-2-4] {$(A_{\lnot\varphi}\otimes A_K)$};
  \end{tikzpicture}
  \caption[]{Example of product of $A_{\lnot\varphi}\otimes A_K$ for $\lnot\varphi = \F(a\land c)\lor \G((\F b)\land (\F\bar b))$ and $K=\F\G(b)\land \G(c)$.
    The dashed transitions are those removed by $\Trim$.  We have $\SG(q_0)=\SG(q_1)=(c)\lor(b\land c)\lor (b\land c)=c$ because these two states can be synchronized with all the states of $A_K$,
     therefore their state guarantee is the disjunction of all labels of $A_K$.
     This result indicates that when the system is synchronized with state $q_0$,
      it will always satisfy $c$.
       We have $\TG(q_1\xrightarrow{\top,\stacc0\stacc1}q_1) = (c)\lor(b\land c) = c$,
        which indicates that when a transition of the system is synchronized with this self-loop,
         it will always satisfy $c$.
         Finally, $\TG(q_0\xrightarrow{\bar a\lor\bar c,\emptyset}q_2)=\bot$ because the only transition synchronizing with this one was trimmed, showing that this transition is not needed.
    \label{fig:construct}}

  \bigskip
  \includegraphics[width=\textwidth]{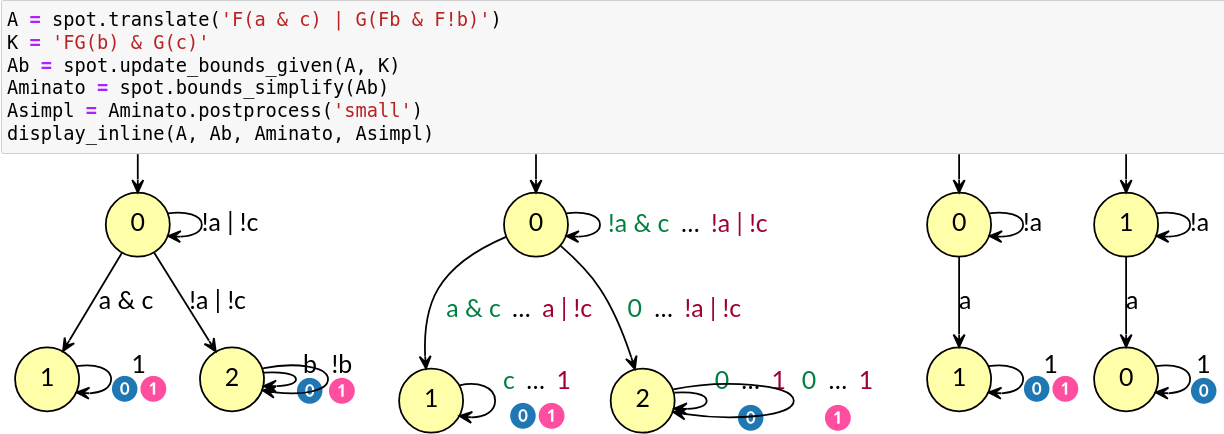}
  \caption{Use of Spot in a Jupyter notebook to integrate some
    knowledge $K=\F\G(b)\land \G(c)$ into the automaton for
    $\lnot\varphi = \F(a\land c)\lor \G((\F b)\land (\F\bar b))$.
    Note that this is the same example as Figure~\ref{fig:construct}
    where the construction of the bounds is explained in detail.
    $A_{\lnot\varphi}$ is on the left.
    The integration of knowledge
    $K$ is represented as an intermediate ``bounded automaton''
    in which each transition is bounded according to
    Theorem~\ref{th:reduc} (second automaton).  Applying Minato's algorithm gives
    the third automaton, which can be further simplified to the
    rightmost automaton reducing the problem to verification of $\F a$.\label{fig:minato}
    Note that Spot's notations differ slightly from those used in the
    paper, for instance \texttt{0}, \texttt{1}, and \texttt{!a|!c}
    stand for $\bot$, $\top$ and $\bar a\lor\bar c$ respectively.}
\end{figure}

In the trim product $\Trim(A_{\lnot \varphi}\otimes
A_K)$, consider a transition $(q,k)\xrightarrow{f\land f_k,a\cup
  a_k}(q',k')$ that was built as a product of
$q\xrightarrow{f,a}q'$ and $k\xrightarrow{f_k,a_k}k'$.
Then it is guaranteed that $f_k\Rightarrow
\SG(q)$ by construction.  Hence, when the component $A_{\lnot
  \varphi}$ of the product is known to be in $q$, this
$\SG(q)$ is an over approximation of the labels
$f_k$ that the transitions in component $A_K$ can satisfy.

Since $A_K$ overapproximates the system
$S$, it is also true that
$\SG(q)$ will overapproximate the behaviors of the states of
$S$ that can synchronize with
$q$.  This state guarantee formula thus provides an upper bound or
over approximation of the system state when reaching
$q$; therefore, for any transition $q\xrightarrow{f,a}q' \in
\delta$, \emph{relaxing} the transition label $f$ to accept $f \lor
\lnot{\SG(q)}$ would not modify the language of the product with the
 system $S$, since the system cannot satisfy
$\lnot{\SG(q)}$ in this state of the product.

Figure~\ref{fig:construct} shows examples of computation of $\SG$.

\subsection{Boolean lower bounds}\label{sec:blb}

Let us look at a way to restrict a transition label of $A_{\lnot\varphi}$
without limiting the ways in which the system can synchronize with
this transition.  For this purpose, we introduce $\TG(t)$ the
\emph{transition guarantee} of a transition $t=q\xrightarrow{f,a}q'$ of
$A_{\lnot\varphi}$, as the disjunction of all labels of transitions of $K$
that synchronize with $t$ in the trim product.

\begin{definition}[Knowledge-based transition guarantee]
  Using the same automata as in Definition~\ref{def:sg}, for any
  transition $t=q\xrightarrow{f,a}q'\in \delta$, we consider the set
  of formulas
  $K_t = \{f_k \mid (q,k)\xrightarrow{f\land f_k,a\cup a_k}(q',k') \in
  \delta_P \}$ that appear on transitions of $A_K$ that synchronize
  with $t$ in the product.
  From the disjunction of this set of formulas, we define the \emph{transition guarantee} for transition $t$:
  \[
    \TG(t) = \bigvee_{f_k\in K_t} f_k
  \]
\end{definition}

Intuitively the label
$f$ of $t$ can be restricted to $f'=f\land \TG(t)$ since this formula
is already enough to match all labels of transitions of $K$ that would synchronize
with $t$ in an accepted run.  Hence, labeling $t$ with $f'$ is also enough to
match all states of the system $S$ that would synchronize with $t$ with
its original label $f$.

Figure~\ref{fig:construct} shows examples of computation of $\TG$.

\subsection{Using the bounds}\label{sec:minato}

\begin{restatable}{theorem}{threduc}\label{th:reduc}
  Using the same automata as in Definition~\ref{def:sg}, consider
  a transition $t=q\xrightarrow{f,a}q'\in \delta$, and let
  $B=\langle \AP, Q, \init, \Acc, \delta\setminus\{t\}\cup\{t'\}\rangle$
be a copy of $A_{\lnot\varphi}$ where $t$ has been replaced by
  $
  t' = q\xrightarrow{f',a}q'
  $
  where $f'\in \B(\AP)$ is any formula such that
  \[
    \underbrace{f\land \TG(t)}_{\text{lower bound}}
      \quad\Rightarrow\quad f'\quad  \Rightarrow\quad
    \underbrace{f\lor \lnot{\SG(q)}}_{\text{upper bound}}
  \]
  Then $\lang(B\otimes A_K)=\lang(A_{\lnot \varphi}\otimes A_K)$
\end{restatable}
For size reasons, the proof is in Appendix~\ref{app:th:reduc}.

Our implementation of this construction is an extension of
Spot~\cite{duret.22.cav} in which the Boolean bounds of
Theorem~\ref{th:reduc} can be represented directly on the automaton.
Figure~\ref{fig:minato} shows our implementation at work.  Our
knowledge bound integration function
\texttt{spot.update\textunderscore{}bounds\textunderscore {}given} can
be called repeatedly to integrate multiple knowledge incrementally, as
we will discuss in Section~\ref{sec:inc}.

To select a simple label compatible with the bounds, we apply Minato's
algorithm~\cite{minato.92.sasimi} (introduced at the top of
Section~\ref{sec:bounds}) to compute a simpler label $f'$ such that
$f\land \TG(t)\implies f' \implies f\lor \lnot{\SG(q)}$.  Note that
when the lower bound is $\bot$, Minato's algorithm will always return
$f'=\bot$ (the transition can be removed), else if the upper bound is
$\top$, $f'=\top$ will be returned.

If $A$ and $K$ are defined over different sets of atomic propositions,
the result of Theorem~\ref{th:reduc} might include atomic propositions
from $K$ that were not in $A$, which is counterproductive.  In this
case, we simplify $K$ by existential quantification of the
propositions that are not in $A$ to produce an automaton $K_{QE}$.
The language of this automaton contains $\lang(K)$, therefore it also
contains $\lang(S)$ and it can still be used as a knowledge.  We denote
$\givenminato{A}{K}$ the ``Bounded by Minato'' automaton, i.e., the
automaton built from $A$ and $K$ by applying Minato's algorithm on the
Boolean bounds computed by Theorem~\ref{th:reduc} with existential
quantification of the atomic propositions not in $A$.

\section{Building Stutter-Insensitive Automata ``Given that...''}\label{sec:stuttu}


When model checking a concurrent system $S$ against a formula $\varphi$, several
advanced and very effective simplification techniques can be used when
it is known that $\lang(A_{\lnot\varphi})$ is stutter-insensitive with
reductions up to a factorial
factor~\cite{peled.94.cav,valmari.93.cav,godefroid.96.phdlncs,haddad.06.ppl,YTM20}.

In this section, we consider the case where the language of
$A_{\lnot\varphi}$ is \emph{stutter-sensitive}, and, given a knowledge
$K$, we want to replace $A_{\lnot\varphi}$ by an automaton $B$ whose
language is \emph{stutter-insensitive}.  Even if $B$ is ``bigger''
than $A_{\lnot\varphi}$, model checking might become more efficient
thanks to the aforementioned simplifications.

Given a word $w$, let $[w]$ be the set of stutter-equivalent words
that can be obtained from $w$ by finitely duplicating letters or removing
repetitions.\footnote{$[w]$ is an equivalence class
  for the $\sim^{\text{lim}}$ relation of Peled et al.~\cite[Lemma
  3]{peled.98.tcs}.}  A language $\lang(A)$ is stutter-sensitive iff
there exists at least one equivalence class $[w]$ that is \emph{only partly}
covered by $\lang(A)$, i.e., such that $[w]\cap \lang(A)\ne\emptyset$
and $[w]\cap \overline{\lang(A)}\ne\emptyset$.

\begin{figure}[tb]
  \includegraphics[width=\textwidth]{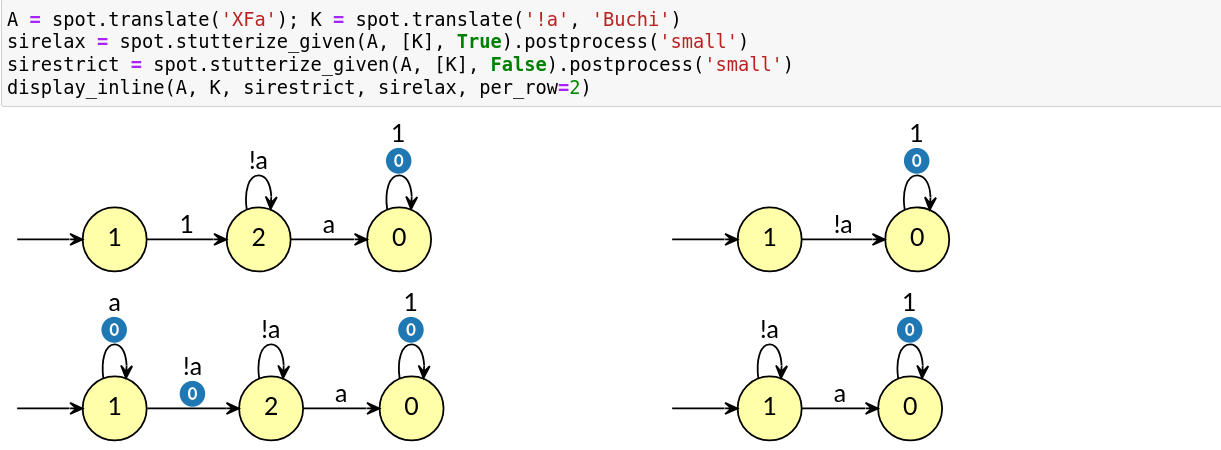}
  \caption{Use of Spot in a Jupyter notebook to integrate some trivial
    knowledge $K=\bar a$ (top right automaton) into the
    stutter-sensitive automaton for $\lnot\varphi = \X\F(a)$ (top-left
    automaton) and turn it into a stutter-insensitive automata.
    Simplified automata for $\givensirestrict{A}{K}$ and
    $\givensirelax{A}{K}$ are given on the bottom left and right
    respectively.\label{fig:stutter}}
\end{figure}

As an example, consider the automaton $A_{\X\F a}$ \emph{given that}
$K=\bar a$.  Figure~\ref{fig:stutter} shows this example using
automata. The knowledge provided here is very simple, but this kind of
effect could occur anywhere in the automaton, not just starting in the
initial state.  Automaton $A_{\X\F a}$ is stutter-sensitive: it
rejects the word $w=a \bar a \bar a\bar a \cdots$, but it accepts
words starting with more than one $a$.  However, notice $w$ is outside
$\lang(K)$, so by Theorem~\ref{th:rel}, $w$ could be added to the
language of $B$ to make it stutter-insensitive (now accepting $\F a$),
giving the bottom-right automaton of Figure~\ref{fig:stutter}.
Another strategy would be to remove all words in $[w]$ from $B$: since
all these words all start by $a$, the whole class $[w]$ is outside
$\lang(K)$.  This second option, corresponding to the formula
$\G(a) \lor \F(\bar a \land \F a)$ gives the bottom-left automaton of
Figure~\ref{fig:stutter}.  Note that on this example, using the
approach based on Minato's algorithm described in
Section~\ref{sec:minato} will only give us bounds $\bar a ... \top$ on
the first transition of $A_{\X\F a}$, but since this transition is
already labeled by $\top$ it would not be changed.

As seen in this example, we propose two strategies to turn a
stutter-sensitive automaton into a stutter-insensitive one.  For each
partly-covered equivalence class $[w]$, the relaxing strategy consists
in adding the rest of $[w]$ to $\lang(A)$.  Dually, the restricting
strategy consists in removing $[w]$ from $\lang(A)$.  This is
legitimated by Theorem~\ref{th:rel} provided that the added or removed
parts are outside the knowledge.

To realize these strategies on automata, let us equip ourselves with a
function $\si(A)$ that returns an automaton $A'$ such that
$\lang(A')$ is the smallest stutter-insensitive language that contains
$\lang(A)$.  Such an operation has already been defined for Büchi
automata~\cite{holzmann.96b.spin} or TGBA~\cite{michaud.15.spin}.
Intuitively, it consists in two simple syntactic transformations: adding
shortcut edges to reduce stutter, and adding states to allow
stuttering after traversing any transition.  The effect of this operation is that all classes
$[w]$ partly covered by $\lang(A)$ get fully included into
$\lang(\si(A))$.

Additionally, let us define $\ss(A)$ the
stutter-sensitive part of $A$ as the automaton that recognizes
only the words $w\in\lang(A)$ such that
$[w]\cap \overline{\lang(A)}\ne\emptyset$.  While, to our knowledge, this
operation does not exist in the literature, it can be defined using
$\si$, complement and product as follows:
\[
  \ss(A) = A\otimes\si(\si(A)\otimes{\overline{A}})
\]
In the above formula, $\si(A)\otimes{\overline{A}}$ accepts exactly the words that should be added to $\lang(A)$ to make it stutter-invariant.  Therefore,
$\si(\si(A)\otimes{\overline{A}})$ accepts all the words $w$ such that
$[w]$ is partly covered by $A$.

We now show how to realize our two strategies using these automata
operations.

\begin{theorem}[Stutter-Insensitive relaxation and restriction]
  \label{th:stutter}
  Let $A$ be a stutter-sensitive TGBA and $K$ be an LTL formula.  We
  define the SI-relaxation and SI-restriction of $A$ given $K$ as
  follows.
\begin{align*}
  \givensirelax{A}{K}&=\begin{cases}
    \si(A) & \text{if~} \lang(\si(A) \otimes \overline{A} \otimes A_K) = \emptyset\\
    A & \text{else}
  \end{cases}\\
  \givensirestrict{A}{K}&=\begin{cases}
                            A\otimes \overline{\ss(A)} & \text{if~} \lang(\ss(A) \otimes A_K) = \emptyset\\
    A & \text{else}
  \end{cases}
\end{align*}

 Then $\lang(\givensirelax{A}{K}\otimes A_K)=\lang(\givensirestrict{A}{K}\otimes A_K)=\lang(A\otimes A_K)$.
\end{theorem}

\begin{proof}
The proof follows from Theorem~\ref{th:rel} and the fact that
$\lang(\ss(A))\subseteq \lang(A) \subseteq \lang(\si(A))$.  Indeed, the
above relaxation returns $\si(A)$ if and only if the words
added by $\si$ (i.e. $\lang(\si(A)\otimes \overline{A})$) are outside
$\lang(K)$ (i.e., $\si(A) \otimes \overline{A} \otimes A_K$ has an
empty language).  Similarly, the restriction removes from $\lang(A)$
the words of $\lang(\ss(A))$ if and only if they are all outside of
$\lang(K)$.  When the given knowledge does not allow adding or
removing those words, the original automaton is returned.\footnote{The
  user of these functions may therefore assume that the returned
  automaton is stutter-insensitive whenever it is different from the
  input.}
\qed
\end{proof}

Figure~\ref{fig:stutter} shows stutter-insensitive automata obtained
with these two constructions.

In practice the complement of $A$ (present in both strategies) can be
avoided when an LTL formula for $A$ is known.  However, the complement
of $\ss(A)$, needed only for $\givensirestrict{A}{K}$ can be rather costly,
especially considering the current definition of $\ss(A)$, which tends
to create large automata.  Fortunately, the latter complementation
need only be performed after it has been checked that the removed
words are not part of the knowledge.  Our hope is therefore that any
stutter-insensitive optimization performed by the model checker will
offset the costs incurred by the computation of $\givensirestrict{A}{K}$.

\section{Incremental Integration of Knowledge}\label{sec:inc}

In this section we show how to
integrate knowledge when we know multiple facts about the system.
We also discuss some strategies to obtain cheap knowledge tailored to help simplify a given property.

\subsection{Working with a Knowledge Base}\label{sec:incremental}

Previously, in Section~\ref{sec:aut}--\ref{sec:stuttu}, we discussed
how to simplify $A_{\lnot\varphi}$ given a single knowledge $K$.
We now assume that we have multiple knowledge facts $A_{K_1}$,
$A_{K_2}$, $\ldots$, $A_{K_n}$ about the system $S$ and discuss
strategies to integrate them all.

We could simplify $A_{\lnot\varphi}$ by applying
Theorems~\ref{th:reduc}--\ref{th:stutter} using
$A_K=A_{K_1}\otimes A_{K_2}\otimes \cdots \otimes A_{K_n}$.  However,
since the product of automata is quadratic in size, this automaton
$A_K$ might be very big.  Even the translation of the conjunction of all
$K_i$ at once $A_K=A_{K_1\land K_2\land \ldots \land K_n}$ might be a
large automaton.

In the following, we propose techniques to integrate knowledge
incrementally, even if this comes with a loss of precision.

For the Boolean Bounds, we suggest applying Theorem~\ref{th:reduc}
using one $A_{K_i}$ at a time, in a loop, and delay the choice of the
label (using Minato's algorithm) until the end of the loop.  In the
syntax of Figure~\ref{fig:minato} we do:
\begin{verbatim}
for k in list_of_facts:
   a = spot.update_bounds_given(a, k)
a_minato = spot.bounds_simplify(a)
\end{verbatim}
Here the operation
\texttt{spot.update\textunderscore{}bounds\textunderscore{}given(a,
  k)} is using the lower bounds of each transition of automaton
\texttt{a} when building the product in the definition of $\TG$ and
$\SG$.  In this loop, each call to
\texttt{spot.update\textunderscore{}bounds\textunderscore{}given} may only
restrict the lower bounds and relax the upper bounds of the
transitions of \texttt{a}.

The incremental construction of $\givensirestrict{A}{K}$ and $\givensirelax{A}{K}$ is handled differently.
Since the automata $\si(A)$ and $A\otimes \overline{\ss(A)}$
constructed by these techniques are independent of the knowledge that
allow to adopt them, we can stop as soon as we find a suitable
knowledge.  More formally:
\begin{align*}
  \givensirelax{A}{K_1,K_2,\ldots,K_n}&=\begin{cases}
    \si(A) & \text{if~} \exists i, \lang(\si(A) \otimes \overline{A} \otimes A_{K_i}) = \emptyset\\
    A & \text{else}
  \end{cases}\\
  \givensirestrict{A}{K_1,K_2,\ldots,K_n}&=\begin{cases}
                            A\otimes \overline{\ss(A)} & \text{if~} \exists i, \lang(\ss(A) \otimes A_{K_i}) = \emptyset\\
    A & \text{else}
  \end{cases}
\end{align*}
This generalization, which is what we implement, explains why
the single fact \texttt{k} was being passed in an array in Figure~\ref{fig:stutter}.
In the implementation the terms $\si(A)\otimes\overline{A}$ and $\ss(A)$ are
of course computed only once, and not for each $K_i$.

\subsection{Seeking Knowledge}
\label{sec:seek}

The strength of our approach is that it is agnostic to the source or
proof method of the knowledge.  Of course some facts might simply be
other formulas we have already proven, when we are dealing with a set
of specification formulas to check against a given system.

We now suggest ways to obtain some cheap knowledge about the system
$S$, tailored to fit the formula $\varphi$ that we intend to verify.
For instance, we can find some simple facts on $S$ using bounded
explorations (breadth-first search, bounded model
checking~\cite{biere.21.bmc}, \ldots), structural analysis of the
system, or a decision procedure for reachability\ldots and that can
help simplify $A_{\lnot\varphi}$.

Our implementation currently looks for various kinds of knowledge, using
simpler decision procedures than full LTL to prove them.
\begin{description}
\item[Initial state]  First, we can check the label of the initial
  state of the system, giving us a knowledge of the form
  $\ell \in 2^\AP$.  While this knowledge is very basic, it is
  free. Using the initial state of the model to simplify a property
  has already been proposed for CTL~\cite{Srba18}.
\item[First steps] Similarly, exploring the first steps of the system
  is cheap.  We compute the set of formulas labeling the transitions
  reachable in the first $n$ steps of $A_{\lnot\varphi}$, and check
  the first $n$ steps of $S$ to check if their values allow
  us to define a knowledge of the form $\X f$, $\X\X f$\ldots

  We limit our exploration to $n=2$ in our experiments.  We use a
  breadth-first search with some limits to avoid explosion on models
  with very large branching factors ($\geq 10^4$). We could also have
  used any technique based on bounded model checking relying on a SAT
  or SMT solver.
\item[Invariants] Proving some invariants of a system can be delegated
  to tools that are specialized in reachability analysis and are more
  effective at this task than LTL model checkers.

  We start by looking at the value of each atomic proposition of
  $\varphi$ in the initial state of $S$ and try to prove that this
  value never evolves using a reachability solver. We then try to
  evaluate compatibility of the atomic propositions, checking given
  two atomic propositions $a$ and $b$ whether all of $\bar a \bar b$,
  $\bar a b$, $a \bar b$ and $ab$ are possible.  For instance,
  $a=[x>2]$ and $b=[x>3]$ have a strong relationship.  Knowledge
  about the exclusions between APs was also used by Blahoudek et
  al.~\cite{blahoudek.15.spin}.  We use an SMT solver to check if some
  of these cases are impossible, not even looking at the system but simply
  at the atomic proposition definitions.  We finally also check if
  formulas labeling the transitions of $A_{\lnot\varphi}$ are
  invariants.  All of these strategies output knowledge of the form
  $\G f$.


\item[Convergent atomic propositions] In this approach we try to prove
  that a given atomic proposition $a$ will eventually converge,
  providing a knowledge of the form $\F (\G a \lor \G \bar a)$.

  We use a low complexity structural test based on an
  analysis of recurring behaviors (SCC in the state graph of the
  system). Any atomic proposition that only observes variables in the
  prefix of such SCC must converge; they cannot oscillate
  indefinitely.  
  We can 
  also, using an SMT solver, try to determine the polarity of
  atomic propositions at convergence, yielding knowledge of the
  form $\F\G a$ (or $\F\G \bar a$).
\end{description}

These strategies are all very basic currently, but show how we can
leverage a diversity of decision procedures (with lower complexity than full LTL model-checking)
 to populate a knowledge
base that is tailored for a given formula to assist an LTL model checking step.

\section{Experimental Study}\label{sec:bench}

Knowledge-based simplifications (``given that'') have been implemented
in Spot~2.13~\cite{duret.22.cav}.  The knowledge collection described
in Section~\ref{sec:seek} has been implemented in
ITS-tools~\cite{ITStools}, which won the LTL category of the Model
Checking Contest in 2023 for the first time, thanks in part to these
strategies. The tools to gather the knowledge and integrate it are
open source and publicly available.  A reproducibility package for the
experiments can be found at
\url{https://codeocean.com/capsule/1210152/tree/v1}.

During the competition, ITS-tools allots a small time slice to
incrementally collect and integrate knowledge.  After this time, it
runs a portfolio of model-checkers, including a symbolic
solution~\cite{ITStools} and LTSmin~\cite{kant.15.tacas} configured as
an explicit model checker with partial-order reductions.

Measurement of the entire model checking procedure would introduce
many biases due to the complex interactions of the portfolio
techniques with the main refinement loop of ITS-Tools~\cite{YTM20}.  Therefore, we focus our evaluation on the knowledge integration step of the procedure, and compare
the automata obtained using the strategies introduced in this
paper.

\subsection{Experimental Setup}

The following performance analysis is based on the models and formulas
of MCC'22~\cite{mcc:2022}.  The benchmark uses a total of $150$
different model families (coming from various domains) configured to
build $1617$ model instances (some models are scalable).

For each of these (colored) Petri net models, the benchmark contains
$32$ randomly generated LTL formulas providing a total of $51744$ LTL
formulas.

For each model instance, we first collected some knowledge using the
basic approach of section~\ref{sec:seek} using
ITS-Tools~\cite{ITStools}, setting a generous timeout of $15$ minutes
to collect it.  Obtaining knowledge is cheap (median $0.67$ minutes
and $75\%$ of cases below $4$ minutes) as it leverages low complexity structural and symbolic tests, and in the worst case reachability queries which are much simpler than full LTL. High time usage to collect this basic knowledge correlates with models where LTL model-checking (at least in the empty product case, with no counter-example) is typically prohibitively expensive (huge models with millions of elements), so that the effort is worth it.

After this processing, we obtain some knowledge for $1601$ model
instances (out of $1617$).  For each model instance the knowledge is
represented as a set of LTL assertions, for a total of $240345$ small
facts (roughly $150$ facts per model instance).  From the original set
of $51744$ LTL formulas we only retain $48975$ formulas that intersect
the gathered facts.


To add some diversity, we consider the above
$48975$ formulas and their negations for a total of
$97950$ formulas.  Note that verifying an LTL formula and its negation
are two independent chalenges: it can be the case that neither of
these formula is verified.


For each formula, we retain only the subset of available facts whose
alphabet intersects that of the formula in the
experiments.\footnote{This is not necessarily optimal.
  Consider $\varphi=\G\F a$ with alphabet $\{a\}$ and facts
  $k_1=\G (b\rightarrow\X a)$ and $k_2=\F\G b$, ignoring $k_2$ because
  its alphabet does not intersect $\varphi$'s is in fact a mistake ;
  however selecting too many facts can easily overload some approaches
  particularly those using the conjunction of known facts, and does
  not help our incremental approaches that consider each fact in
  isolation since they typically ignore atomic propositions not in
  $\varphi$.}

Our benchmark therefore contains $97950$ problems that consist in one
specification LTL formula accompanied by a set of knowledge facts
(on average $12.7$, median $9$ facts per formula).

We then proceed to apply each of our strategies to these problems to build an automaton and compute various metrics on its size.
The strategies we compare are the following.  A ``p.'' used as prefix
indicates a \emph{precise} construction that considers all facts
$K=\bigwedge_i K_i$ at once. While \emph{precise} variants can pay a significant cost to manipulate
the conjunction of known facts, they also benefit from a more precise
knowledge so that there is a trade-off between precise and incremental
approaches.

\begin{description}
\item[raw] is formula $\lnot\varphi$ translated to a TGBA, without any integration of knowledge
\item[p.min, p.max] are obtained by building $A_{\lnot\varphi\land K}$ and
  $A_{\lnot\varphi\lor\lnot K}$ as discussed in Section~\ref{sec:basic},
  where $K$ is the conjunction $K$ of all facts.
\item[p.min$\exists$, p.max$\exists$] are the variants with existential quantification shown in equations \eqref{eq:minqe}--\eqref{eq:maxqe} from Section~\ref{sec:basic}.
\item[p.BM, BM] use respectively the strategy $\givenminato{A}{K}$ presented in Section~\ref{sec:minato}, and the incremental strategy described in Section~\ref{sec:incremental}.
\item[p.SIrelax, p.SIrestrict, SIrelax, SIrestrict] are the strategies of Section~\ref{sec:stuttu}, and their incremental variants from Section~\ref{sec:incremental}
\end{description}
Finally, we also consider some combination of techniques.  For instance
``SIrelax+BM'' designates the incremental implementation of SIrelax followed
by the incremental implementation of BM.

When providing statistics about the automata produced by the above
variants, we always assume that those automata have been further
simplified using techniques implemented in Spot (notably, removing
useless states, useless acceptance marks, and using simulation-based
reductions to merge states and prune unnecessary
transitions~\cite{babiak.13.spin}).

The average runtime for solving a problem with any strategy is 35.6ms.
On the 97950 benchmark problems, the only strategies that exceed a
very generous timeout of 10 seconds are ``SIrestrict'' on 470
problems, and ``p.SIrestrict'' on 522 problems.  Those strategies are
occasionally very slow only because of the amount of automata
complementations they have to perform.  Overall the knowledge integration step is truly negligible before any test involving the actual system. If the knowledge gathering step only uses low complexity procedures or knowledge  simply consists of previously proven properties, knowledge integration scales exceptionally well to complex problems.

\subsection{Problems Reduced to Empty or Universal}\label{sec:univempty}

We first study the problems that could be fully solved given the
knowledge by reducing the automaton to an empty or universal one.  For
testing universality, we syntactically check if the resulting
automaton has been reduced to a single-state all-accepting automaton.

\begin{table}[tb]
  \caption{Amount of problems (out of 97950 composed of formulas and their negation) that could be shown to be empty or universal using the provided knowledge.\label{tab:univempty}}
  ~\hfill
\begin{tabular}{lrrr}
\toprule
  Strategy&Universal&Empty&Total\\
  \midrule
p.min&0&25508&25508\\
p.max&24453&0&24453\\
p.min$\exists$&0&25508&25508\\
p.max$\exists$&25508&0&25508\\
BM&23080&25095&48175\\
SIrelax&208&0&208\\
SIrestrict&0&219&219\\
  \bottomrule
\end{tabular}
  \hfill
\begin{tabular}{lrrr}
\toprule
  Strategy&Universal&Empty&Total\\
  \midrule
p.BM&23258&25508&48766\\
p.SIrelax&212&0&212\\
p.SIrestrict&0&223&223\\
SIrelax+BM&23286&25091&48377\\
BM+SIrelax&23164&25095&48259\\
p.SIrelax+p.BM&23708&25508&49216\\
  p.BM+p.SIrelax&23344&25508&48852\\
  \bottomrule
\end{tabular}
  \hfill~
\end{table}

Table~\ref{tab:univempty} presents those results.
While it is certainly due to the random nature of the formulas of
 the MCC, in total $51016/97950\approx 52\%$ of the formulas of
  the MCC benchmark we kept ($49\%$ of all formulas of the MCC)
   could be solved using only the basic approach to glean related knowledge presented in Section~\ref{sec:seek}, thus avoiding a full LTL model-checking procedure.

We can see that
``min$\exists$'' and ``p.min$\exists$'' are the most effective
strategies to find empty problems, on par with ``p.BM''.  Dually
``p.max$\exists$'' is the only most effective at deducing universal
problems.  The amount of problems reduced to empty by
``p.min$\exists$'' and to universal by ``p.max$\exists$'' are
identical as hoped, because our benchmark includes both formulas and their
negations. Strategy ``p.max'' is less effective than ``p.max$\exists$'' because it keeps atomic propositions that are not in $\varphi$.
Still, while they might produce a larger automaton as discussed in the next section if they can't solve the problem, it is important in a full decision approach involving some knowledge to first test ``p.min$\exists$'' and
``p.max$\exists$'' for full solutions.

Generally, strategies based only on restriction (resp. relaxation) can
prove only emptiness (resp. universality) and obtaining ``empty'' seems easier on this benchmark than obtaining ``universal'' perhaps due to our limited syntactic check for universality.

\subsection{Simplifying the Remaining Unsolved Formulas}

\begin{table}[tb]
  \caption{\label{tab:strategies}Comparison of the different
    strategies over 46934 problems that could not be already reduced
    to false or true by previous methods.  `raw' designates the
    original automata, for baseline. For each strategy we report
    various metrics of the produced automata: its number of states and
    transitions, $\sum|f|$ is the total size of all labels, SI
    (resp. det) shows the fraction of automata that were
    stutter-insensitive (resp. deterministic), $|AP|$ is the number of
    atomic proposition, Time reports the number of milliseconds needed
    by the strategy, and TO counts the number of timeouts ($>10$
    seconds).  Different statistics are provided for some
    measurements: `q95' denotes the 95\% quantile (i.e., 95\% of all
    values are below the indicated value), `geom' denotes the
    geometric mean.  Values within 2\% of the best (resp. worse) value
    of a column, `raw` excluded, are highlighted in
    \colorbox{yellow}{yellow} (resp. \colorbox{pink}{pink}).}
  \centering\resizebox{\linewidth}{!}{
\begin{tabular}{l>{}c>{}c>{}c>{}c>{}c>{}c>{}c>{}c>{}c>{}c>{}c>{}c>{}c>{}c>{}c}
\toprule
 & \multicolumn{4}{c}{States} & \multicolumn{4}{c}{Transitions} & $|AP|$ & $\sum|f|$ & SI & det & \multicolumn{2}{c}{Time (ms)} & \multicolumn{1}{c}{} \\ \cmidrule(lr){2-5}\cmidrule(lr){6-9}\cmidrule(lr){14-15}
Strategy  & q95 & max & mean & geom & q95 & max & mean & geom & mean & mean &  &  & mean & geom & \multicolumn{1}{c}{TO} \\
\midrule
raw  & \cellcolor{white}{$\phantom{0}8$} & \cellcolor{white}{$\phantom{000}73$} & \cellcolor{white}{$\phantom{0}3.71$} & \cellcolor{white}{$3.13$} & \cellcolor{white}{$18$} & \cellcolor{white}{$\phantom{0000}286$} & \cellcolor{white}{$\phantom{00}7.32$} & \cellcolor{white}{$\phantom{0}5.52$} & \cellcolor{white}{$2.14$} & \cellcolor{white}{$\phantom{0}21.76$} & \cellcolor{white}{49\%} & \cellcolor{white}{50\%} & \cellcolor{white}{$20.26$} & \cellcolor{white}{$19.95$} & \cellcolor{white}{$\phantom{00}0$ }\\
\midrule
p.min  & \cellcolor{white}{$\phantom{0}9$} & \cellcolor{white}{$\phantom{00}684$} & \cellcolor{white}{$\phantom{0}5.30$} & \cellcolor{white}{$4.79$} & \cellcolor{white}{$19$} & \cellcolor{white}{$\phantom{00}16834$} & \cellcolor{white}{$\phantom{00}8.71$} & \cellcolor{white}{$\phantom{0}6.71$} & \cellcolor{pink}{$3.35$} & \cellcolor{white}{$\phantom{0}48.36$} & \cellcolor{pink}{9\%} & \cellcolor{white}{50\%} & \cellcolor{white}{$55.04$} & \cellcolor{white}{$49.82$} & \cellcolor{yellow}{$\phantom{00}0$ }\\
p.max  & \cellcolor{white}{$10$} & \cellcolor{white}{$\phantom{000}76$} & \cellcolor{white}{$\phantom{0}6.30$} & \cellcolor{pink}{$5.90$} & \cellcolor{white}{$25$} & \cellcolor{white}{$\phantom{0000}295$} & \cellcolor{white}{$\phantom{0}12.78$} & \cellcolor{pink}{$11.23$} & \cellcolor{pink}{$3.35$} & \cellcolor{white}{$\phantom{0}58.15$} & \cellcolor{pink}{9\%} & \cellcolor{pink}{43\%} & \cellcolor{white}{$55.86$} & \cellcolor{white}{$50.27$} & \cellcolor{yellow}{$\phantom{00}0$ }\\
p.min$\exists$  & \cellcolor{white}{$\phantom{0}9$} & \cellcolor{white}{$\phantom{00}684$} & \cellcolor{white}{$\phantom{0}5.23$} & \cellcolor{white}{$4.67$} & \cellcolor{white}{$19$} & \cellcolor{white}{$\phantom{00}16834$} & \cellcolor{white}{$\phantom{00}8.60$} & \cellcolor{white}{$\phantom{0}6.55$} & \cellcolor{white}{$2.21$} & \cellcolor{white}{$\phantom{0}35.06$} & \cellcolor{white}{11\%} & \cellcolor{white}{50\%} & \cellcolor{white}{$54.18$} & \cellcolor{white}{$48.97$} & \cellcolor{yellow}{$\phantom{00}0$ }\\
p.max$\exists$  & \cellcolor{white}{$10$} & \cellcolor{white}{$\phantom{000}76$} & \cellcolor{white}{$\phantom{0}6.14$} & \cellcolor{white}{$5.73$} & \cellcolor{white}{$25$} & \cellcolor{white}{$\phantom{0000}295$} & \cellcolor{white}{$\phantom{0}12.37$} & \cellcolor{white}{$10.77$} & \cellcolor{white}{$2.21$} & \cellcolor{white}{$\phantom{0}42.54$} & \cellcolor{white}{11\%} & \cellcolor{white}{45\%} & \cellcolor{white}{$54.91$} & \cellcolor{white}{$49.39$} & \cellcolor{yellow}{$\phantom{00}0$ }\\
BM  & \cellcolor{yellow}{$\phantom{0}6$} & \cellcolor{yellow}{$\phantom{000}65$} & \cellcolor{white}{$\phantom{0}3.13$} & \cellcolor{white}{$2.68$} & \cellcolor{yellow}{$13$} & \cellcolor{yellow}{$\phantom{0000}286$} & \cellcolor{yellow}{$\phantom{00}5.42$} & \cellcolor{white}{$\phantom{0}4.24$} & \cellcolor{yellow}{$1.70$} & \cellcolor{yellow}{$\phantom{0}13.39$} & \cellcolor{white}{46\%} & \cellcolor{yellow}{59\%} & \cellcolor{yellow}{$41.46$} & \cellcolor{yellow}{$40.62$} & \cellcolor{yellow}{$\phantom{00}0$ }\\
SIrelax  & \cellcolor{white}{$\phantom{0}8$} & \cellcolor{white}{$\phantom{000}73$} & \cellcolor{white}{$\phantom{0}3.86$} & \cellcolor{white}{$3.16$} & \cellcolor{white}{$22$} & \cellcolor{yellow}{$\phantom{0000}286$} & \cellcolor{white}{$\phantom{00}8.07$} & \cellcolor{white}{$\phantom{0}5.82$} & \cellcolor{white}{$2.14$} & \cellcolor{white}{$\phantom{0}26.51$} & \cellcolor{white}{66\%} & \cellcolor{white}{49\%} & \cellcolor{yellow}{$42.19$} & \cellcolor{yellow}{$41.09$} & \cellcolor{yellow}{$\phantom{00}0$ }\\
SIrestrict  & \cellcolor{white}{$10$} & \cellcolor{white}{$19463$} & \cellcolor{white}{$\phantom{0}6.77$} & \cellcolor{white}{$3.23$} & \cellcolor{white}{$27$} & \cellcolor{white}{$\phantom{0}373093$} & \cellcolor{white}{$\phantom{0}52.25$} & \cellcolor{white}{$\phantom{0}5.96$} & \cellcolor{white}{$2.14$} & \cellcolor{white}{$335.58$} & \cellcolor{white}{67\%} & \cellcolor{white}{51\%} & \cellcolor{white}{$57.65$} & \cellcolor{white}{$44.28$} & \cellcolor{white}{$122$ }\\
p.BM  & \cellcolor{yellow}{$\phantom{0}6$} & \cellcolor{yellow}{$\phantom{000}65$} & \cellcolor{white}{$\phantom{0}3.13$} & \cellcolor{white}{$2.68$} & \cellcolor{yellow}{$13$} & \cellcolor{yellow}{$\phantom{0000}286$} & \cellcolor{yellow}{$\phantom{00}5.41$} & \cellcolor{white}{$\phantom{0}4.24$} & \cellcolor{yellow}{$1.69$} & \cellcolor{yellow}{$\phantom{0}13.35$} & \cellcolor{white}{46\%} & \cellcolor{yellow}{59\%} & \cellcolor{white}{$46.13$} & \cellcolor{white}{$43.50$} & \cellcolor{yellow}{$\phantom{00}0$ }\\
p.SIrelax  & \cellcolor{white}{$\phantom{0}9$} & \cellcolor{white}{$\phantom{000}73$} & \cellcolor{white}{$\phantom{0}3.93$} & \cellcolor{white}{$3.18$} & \cellcolor{white}{$24$} & \cellcolor{white}{$\phantom{0000}340$} & \cellcolor{white}{$\phantom{00}8.47$} & \cellcolor{white}{$\phantom{0}5.92$} & \cellcolor{white}{$2.14$} & \cellcolor{white}{$\phantom{0}29.15$} & \cellcolor{yellow}{70\%} & \cellcolor{white}{49\%} & \cellcolor{white}{$46.20$} & \cellcolor{white}{$43.71$} & \cellcolor{yellow}{$\phantom{00}0$ }\\
p.SIrestrict  & \cellcolor{pink}{$11$} & \cellcolor{pink}{$84249$} & \cellcolor{pink}{$10.59$} & \cellcolor{white}{$3.28$} & \cellcolor{pink}{$32$} & \cellcolor{pink}{$1252969$} & \cellcolor{pink}{$105.47$} & \cellcolor{white}{$\phantom{0}6.12$} & \cellcolor{white}{$2.14$} & \cellcolor{pink}{$677.23$} & \cellcolor{yellow}{70\%} & \cellcolor{white}{51\%} & \cellcolor{white}{$57.34$} & \cellcolor{white}{$45.53$} & \cellcolor{pink}{$130$ }\\
SIrelax+BM  & \cellcolor{yellow}{$\phantom{0}6$} & \cellcolor{yellow}{$\phantom{000}65$} & \cellcolor{yellow}{$\phantom{0}3.07$} & \cellcolor{yellow}{$2.62$} & \cellcolor{yellow}{$13$} & \cellcolor{yellow}{$\phantom{0000}286$} & \cellcolor{yellow}{$\phantom{00}5.38$} & \cellcolor{yellow}{$\phantom{0}4.18$} & \cellcolor{yellow}{$1.70$} & \cellcolor{yellow}{$\phantom{0}13.49$} & \cellcolor{white}{51\%} & \cellcolor{yellow}{59\%} & \cellcolor{white}{$63.37$} & \cellcolor{white}{$61.75$} & \cellcolor{yellow}{$\phantom{00}0$ }\\
BM+SIrelax  & \cellcolor{white}{$\phantom{0}7$} & \cellcolor{yellow}{$\phantom{000}65$} & \cellcolor{white}{$\phantom{0}3.19$} & \cellcolor{white}{$2.67$} & \cellcolor{white}{$15$} & \cellcolor{yellow}{$\phantom{0000}286$} & \cellcolor{white}{$\phantom{00}5.93$} & \cellcolor{white}{$\phantom{0}4.49$} & \cellcolor{yellow}{$1.70$} & \cellcolor{white}{$\phantom{0}16.21$} & \cellcolor{white}{67\%} & \cellcolor{yellow}{58\%} & \cellcolor{white}{$63.15$} & \cellcolor{white}{$61.46$} & \cellcolor{yellow}{$\phantom{00}0$ }\\
p.SIrelax+p.BM  & \cellcolor{yellow}{$\phantom{0}6$} & \cellcolor{yellow}{$\phantom{000}65$} & \cellcolor{yellow}{$\phantom{0}3.04$} & \cellcolor{yellow}{$2.59$} & \cellcolor{yellow}{$13$} & \cellcolor{yellow}{$\phantom{0000}286$} & \cellcolor{yellow}{$\phantom{00}5.34$} & \cellcolor{yellow}{$\phantom{0}4.15$} & \cellcolor{yellow}{$1.69$} & \cellcolor{yellow}{$\phantom{0}13.44$} & \cellcolor{white}{52\%} & \cellcolor{yellow}{59\%} & \cellcolor{pink}{$72.08$} & \cellcolor{pink}{$66.89$} & \cellcolor{yellow}{$\phantom{00}0$ }\\
p.BM+p.SIrelax  & \cellcolor{white}{$\phantom{0}7$} & \cellcolor{yellow}{$\phantom{000}65$} & \cellcolor{white}{$\phantom{0}3.16$} & \cellcolor{yellow}{$2.63$} & \cellcolor{white}{$16$} & \cellcolor{yellow}{$\phantom{0000}286$} & \cellcolor{white}{$\phantom{00}5.99$} & \cellcolor{white}{$\phantom{0}4.46$} & \cellcolor{yellow}{$1.69$} & \cellcolor{white}{$\phantom{0}16.83$} & \cellcolor{yellow}{70\%} & \cellcolor{yellow}{58\%} & \cellcolor{pink}{$71.85$} & \cellcolor{pink}{$66.69$} & \cellcolor{yellow}{$\phantom{00}0$ }\\
\bottomrule
\end{tabular}}
\end{table}

We now study in Table~\ref{tab:strategies} statistics for all the 46934 problems
that could not be proven empty or universal by any strategy.

Since our goal is to reduce the size of the automaton, we first study
the number of states and transitions.  To better understand the
distribution of values, we present the 95\% quantile, as well as the
arithmetic and geometric means.  Cases where the arithmetic mean is
much larger than the geometric mean indicate the presence of a few
very large outliers.

We observe that basic strategies based on ``p.min'' or ``p.max'' are
not very good, as feared, doubling the average size.  However, all
strategies involving ``BM'' perform well: the average number of state
is reduced by 15\%, and transitions by 25\%.  The ``SIrelax'' strategy
produces a moderate size increase that can be further alleviated by
combining it with ``BM''.  However, ``SIrestrict'' can dramatically
increase the size of the automaton, and even time out in extreme
cases.

The number of atomic propositions (column $|AP|$) and size of the
formula labels ($\sum|f|$) is also significantly reduced by all
variants using ``BM''.  The average number of atomic propositions is
reduced from 2.14 to around 1.7 (a 21\% gain), and the average size of
formula labels goes from 21.76 to around 13.4 (a 38\% gain).

Concerning the stutter insensitivity (column ``SI'') of the resulting
automaton, only $49\%$ of the ``raw'' problems are
stutter-insensitive. ``min'' and ``max'' degrade this number
significantly. Both of the strategies ``SIrestrict'' and ``SIrelax''
developed to optimize this metric are indeed effective (but
``SIrestrict'' is more expensive and liable to timeout). Combined
strategies that finish with a ``SIrelax'' step lead to the best
results, being both small and stutter-insensitive in $70\%$ of cases.
The precise variants are a bit better than the incremental
constructions.

While our algorithm is not looking to improve determinism (column
``det''), this characteristic is nonetheless improved by all variants
involving ``BM''.  This is a welcome side-effect since having small
\emph{and} deterministic automata can only help model
checking~\cite{sebastiani.03.charme,blahoudek.14.spin}.

On this subset of 46934 cases, the average time to solve a problem is
up to two times higher than the average of the 97950 benchmark
problems (which was 35.6ms as mentioned earlier).  However it is still
very cheap.  The fastest strategies to integrate knowledge is ``BM''
(with an average of $41$ms). ``p.'' precise variants all pay a
reasonable time penalty, but the improvement in size is very modest.
The only timeouts we observe on these problems that cannot be entirely
solved are for ``SIrestrict'' and its precise variant (in less than
$3$\textperthousand~cases).

In conclusion, combined strategies using ``SIrelax'' and ``BM''
produce the smallest automata without real drawbacks apart from
moderate increase of the run time.  Given the very reasonable run
times, it is even feasible to run several of these strategies and then
select the most appropriate automaton on a case by case basis.

\section{Related Work}

Theorem~\ref{th:rel} proposes an original framework for exploiting
prior knowledge in LTL model checking.  This generalizes approaches
that only consider invariants~\cite{Srba18} or
quasi-invariants~\cite{larraz.14.cav}.

Theorem~\ref{th:rel} is also related to the problem of language
separation: given two languages, a separator is a third language that
contains the first one and is disjoint from the second
one~\cite{place.16.lmcs}.  In our case, we are looking for an
automaton $B$ whose language separates $\lang(A)\cap\lang(K)$ (which
it should include) from $\lang(K)\setminus\lang(A)$ (which it should
not intersect).  However, the two languages to separate aren't
independent: $A$ is already known to be a separator, and we are trying
to find a simpler $B$ by simplifying $A$.

Blaoudek et~al.~\cite[Section~5]{blahoudek.15.spin} also consider a
simplification of labels leveraging Minato's algorithm as we did in
Section~\ref{sec:bounds}.  While it is limited to an invariant about
mutually exclusive propositions, it did prove to be an effective
simplification.  Our approach generalizes theirs: if the knowledge
encodes mutual exclusion of atomic propositions, we will generate the
same bounds, however we can handle arbitrary LTL knowledge, and we
take the structure of the automaton into account.

Using Minato's algorithm to find a simple $f'$ such that
$f_{low}\implies f' \implies f_{high}$ can be related to Coudert and
Madre's \texttt{restrict} and \texttt{constraint}
operators~\cite{coudert.90.iccad} that find $f'$ such that
$f\land c \implies f'\implies f\lor\lnot c$, where $c$ is a Boolean
formula. However, in our case, $f_{low}$ and $f_{high}$ are not
limited to this form.

The use of bounded automata in Section~\ref{sec:bounds} evokes the
notion of incompletely specified Mealy machines used in synthesis,
where ``don't care'' edges are leveraged to produce smaller
automata~\cite{paull.59.tec,abel.15.iccad,renkin.22.forte}.  The
bounded automata we propose can be used for bound-aware
simulation-based reductions~\cite{smolka.23.bt}; this could complement
our current approaches.

Dureja and Rozier~\cite{dureja.18.tacas} consider the problem of model
checking a single model against a large set of LTL formulas.  They
compute a matrix of implications between formulas $f_i\implies f_j$,
and they use previously proven formula $f_1$ to avoid model checking
of implied formulas $f_2$.
Such an implication test, between a previously proven formula $f_1$
(the knowledge) and an unproved formula $f_2$, is covered in our
approach since ``$f_2$ given $f_1$'' will be an empty automaton (see
Section~\ref{sec:univempty}).  However, we can also obtain a simpler
automaton for $f_2$ even in the absence of full implication.
Moreover, we suggest several approaches to leverage \emph{all} accumulated
knowledge incrementally.

Our definitions suggest explicit representation of automata, however
our approach can be used for symbolic model checking.  Instead of
using a direct symbolic encoding of Büchi
automaton~\cite{rozier.11.fm}, obtained directly from LTL, we can
encode the explicit automaton resulting from our knowledge
simplifications into a symbolic
representation~\cite{sebastiani.05.cav}.  In fact, ITS-tools
uses both a knowledge-based approach and a symbolic encoding.


Although this work is motivated by model checking, our techniques can
be used to optimize any inclusion check $\lang(A)\subseteq\lang(B)$.
E.g., in the traditional implementation based on a
complementation~\cite{tsai.14.lmcs} of $B$, any knowledge about $A$,
can be used to simplify $B$ before its complementation.

\section{Conclusion}

We have introduced new operations that help simplify the
model-checking of a new formula when we already possess some prior
knowledge on the system.  Our strategies are automata-based
operations, thus capturing any nature of LTL property or prior
knowledge.  The evaluation of our current implementation on a large
benchmark demonstrates the effectiveness of the approach.

Studying the problem of knowledge integration led us to the problem of
producing a (small) automaton given bounds on the language it
represents.  This challenging problem is new to our knowledge and
while we have proposed several strategies in this paper, there is a
lot of room for more research in this direction.  For instance the
strategies we presented in Section~\ref{sec:stuttu} to produce
stutter-insensitive automata currently do not take any advantage of
the Boolean bounds computed in Section~\ref{sec:bounds}.  Similarly,
those Boolean bounds could very likely be used for other kinds of
simplifications, such as bound-aware simulation-based
reductions~\cite{smolka.23.bt}.
The problem of seeking relevant knowledge by leveraging simpler decision procedures than full LTL is also an avenue for further exploration, paving the way to
strategies achieving an incremental verification process.






\bibliographystyle{splncs04}
\bibliography{biblio,mc}

\newpage

The following appendix is not meant to be part of the published paper
because of size restrictions.  It is included for the benefit of the
interested reviewers.

\appendix
\section{Proof of Theorem~\ref{th:reduc}}\label{app:th:reduc}

\threduc*

\begin{proof}
  As a preliminary, notice that since $B$ and $A_{\lnot\varphi}$ differ
  only by the labels of their transitions, but the product $B\otimes A_K$
  and $A_{\lnot \varphi}\otimes A_K$ have the same states.

  $(\supseteq)$  Consider a word
  $w=\ell_0\ell_1\ell_2\ldots\in\lang(A_{\lnot \varphi}\otimes A_K)$.
  There exists an accepting run
  $r=(q_1,q_{k1})\xrightarrow[t_1]{f_1\land f_{k1},a_1\cup a_{k1}}
  (q_2,q_{k2})\xrightarrow[t_2]{f_2\land f_{k2},a_2\cup a_{k2}}\cdots$
  of $A_{\lnot \varphi}\otimes A_K$ such that for all $i$, we have $\ell_i\implies f_i\land f_{ki}$.
  We named the transitions $t_1$, $t_2$, ... for later reference.
  By definition of the product, this run can be seen a the synchronization
  of two runs:
  $r_A=q_1\xrightarrow[t_{a1}]{f_1,a_1}q_2\xrightarrow[t_{a2}]{f_2,a_2}\cdots$ a run of $A_{\lnot\varphi}$ accepting $w$,
  and
  $r_K=q_{k1}\xrightarrow{f_{k1},a_{k1}} (q_{k2})\xrightarrow{f_{k2},a_{k2}}\cdots$ a run of
  $A_K$ accepting $w$ as well.

  Since $B$ has been constructed from $A_{\lnot\varphi}$ by just
  changing the labels of the edges to anything permitted by the
  theorem, $B$ necessarily contains a run
  $r_B=q_1\xrightarrow{f'_1,a_1}q_2\xrightarrow{f'_2,a_2}\cdots$ such
  that $f_i\land \TG(t_{ai}) \implies f_i'$ for each $i$.  This run $r_B$ is accepting because
  it sees the same acceptance marks as $r_A$.   We claim that $r_b$
  is an accepting run on $w$
  because it can be shown that $\ell_i\implies f_i\land \TG(t_{ai}) \implies f_i'$.

  Since it is already the case that $\ell_i\implies f_i$ for each $i$
  (since $r_A$ accepts $w$), we just have to prove that
  $\ell_i\implies\TG(t_{ai})$.  The transition $t_i$ (which was
  obtained by synchronizing $t_{ai}$ and $t_{ki}$) is part of the
  accepting run $r$, so it also belongs to
  $\Trim(A_{\lnot\varphi}\otimes A_K)$.  This means that $\TG(t_{ai})$
  contains at least $f_{ki}$ as a disjunct.  We can therefore say that
  $\ell_i\implies f_{ki}\implies \TG(t_{ai})$ for all $i$.
  Conclusion: $w$ is still accepted by $B$ and therefore by $B\otimes A_K$ as well.

  ($\subseteq$) Consider an accepted word
  $w=\ell_0\ell_1\ell_2\ldots\in\lang(B\otimes A_K)$.
  There exists an accepting run
  $r'=(q_1,q_{k1})\xrightarrow[t_1']{f_1'\land f_{k1},a_1\cup a_{k1}}
  (q_2,q_{k2})\xrightarrow[t_2']{f_2'\land f_{k2},a_2\cup a_{k2}}\cdots$
  of $B\otimes A_K$ such that for all $i$, we have $\ell_i\implies f_i'\land f_{ki}$.

  By definition of the product, this run can be seen a the
  synchronization of two runs:
  $r_B=q_1\xrightarrow[t_{b1}]{f_1',a_1}q_2\xrightarrow[t_{b2}]{f_2',a_2}\cdots$
  a run of $B$, and
  $r_K=q_{k1}\xrightarrow{f_{k1},a_{k1}}
  q_{k2}\xrightarrow{f_{k2},a_{k2}}\cdots$ a run of $A_K$, both
  accepting $w$.

  Because of the way $B$ has been constructed in the Theorem, we know
  that for each transition $t_{bi}$ there is a corresponding transition
  $q_i\xrightarrow{f_i,a_i}q_{i+1}$ in $A_{\lnot\varphi}$ such that
  $f_i'\implies f_i\lor\lnot\SG(q_i)$.  Therefore, since
  $\ell_i\implies f_i'$ we have
  $(\ell_i\implies f_i)\lor(\ell_i\implies \lnot\SG(q_i))$.
  Let us show that the latter clause is false, so that
  the former one needs to be true.

  Since $\SG(q_i)$ is the disjunction of all labels that $A_K$ could
  do when it is synchronized with $q_i$, let us \emph{assume} that
  $(q_i,q_{ki})$ is reachable in $\Trim(A_{\lnot\varphi}\otimes A_K)$.  Then
  the transition $q_{ki}\xrightarrow{f_{ki},a_{ki}}q_{k(i+1)}$ is
  considered when building the disjuncts of $\SG(q_i)$, so we have
  $\ell_i\implies f_{ki} \implies \SG(q_i)$.  This implies that
  $\ell_i \centernot\implies \lnot\SG(q_i)$, which, combined with the
  last equation of the previous paragraph implies that
  $\ell_i\implies f_i$.

  We conclude that if $(q_i,q_{ki})$ is reachable, then not only
  $\ell_i\implies f_i$, but also the transition
  $(q_i,q_{ki})\xrightarrow{f_i\land f_{ki},a_i\cup a_{ki}}
  (q_{i+1},q_{k(i+1)})$ exists in $A_{\lnot\varphi}\otimes A_K$ making
  $(q_{i+1},q_{k(i+1)})$ reachable in turn.

  Since the initial state is obviously reachable, the above
  reasoning allows us to inductively define a run
$r=(q_1,q_{k1})\xrightarrow{f_1\land f_{k1},a_1\cup a_{k1}}
(q_2,q_{k2})\xrightarrow{f_2\land f_{k2},a_2\cup a_{k2}}\cdots$
of $A_{\lnot\varphi}\otimes A_K$  that accepts $w$.
  \qed
\end{proof}

\end{document}